\documentclass[12pt,a4paper]{iopart}
\usepackage{graphicx,setstack}
\usepackage{amsgen, amsopn, amsthm, amsfonts, amscd, amssymb}
\usepackage{mathrsfs}

\DeclareMathOperator{\reg}{reg} 

\DeclareMathOperator{\Spin}{Spin}

\theoremstyle{plain}
\newtheorem{thm}{Theorem}[section]
\newtheorem{lem}[thm]{Lemma}
\newtheorem{pro}[thm]{Proposition}

\theoremstyle{definition}

\theoremstyle{remark}
\newtheorem{egs}[thm]{Example}

\newtheorem{rem}[thm]{Remark}

\newcommand{\CC}{{\mathbb C}}

\newcommand{\RR}{{\mathbb R}}

\newcommand{\ZZ}{{\mathbb Z}}

\newcommand{\dt}{d_{S^1}}

\newcommand{\set}[1]{\left\{ #1 \right\}}

\newcommand{\dfrac}{\displaystyle\frac}
\newcommand{\Hy}{H^{\bullet}}

\begin{document}

\title[The $\hat{\Gamma}$-genus]{The $\hat{\Gamma}$-genus and a regularization of an $S^1$-equivariant Euler class}
\author{Rongmin Lu}
\address{School of Mathematical Sciences, University of Adelaide, South Australia 5005, Australia}
\ead{rongmin.lu@adelaide.edu.au}
\date{\today}
\begin{abstract}
We show that a new multiplicative genus, in the sense of Hirzebruch, can be obtained by generalizing a calculation due to Atiyah and Witten. We introduce this as the $\hat{\Gamma}$-genus, compute its value for some examples and highlight some of its interesting properties. We also indicate a connection with the study of multiple zeta values, which gives an algebraic interpretation for our proposed regularization procedure.
\end{abstract}

\section{Introduction}

The Atiyah--Singer index theorem \cite{AtiyahSinger68} was a cornerstone of 20th century mathematics, setting off an interplay of ideas between mathematics and physics that has continued till this day. Among the many subsequent re-derivations of this theorem was a formal demonstration by Atiyah \cite{Atiyah85}, following an idea of Witten \cite{Witten85}, of the equivalent result that the $\hat{A}$-genus of a spin manifold $M$ can be recovered as a regularized $S^1$-equivariant Euler characteristic of the normal bundle of $M$ in its free loop space $LM$.

In this note, we extend the regularization of Atiyah and Witten to the case of an arbitrary complex vector bundle $\pi\,\colon E\to M$ of rank $m\geq 2$ with spin structure (here, $M$ need not be spin). We find that we derive a new multiplicative genus, in the sense of Hirzebruch \cite{Hirzebruch95}, and we introduce this as the $\hat{\Gamma}$-genus (see Proposition \ref{pro:invEuler}).

The $\hat{\Gamma}$-genus has several interesting properties: for example, it vanishes on all Riemann surfaces (see Proposition \ref{pro:gamma-hat}). It is the second multiplicative genus that incorporates the $\Gamma$ function into its generating function: the first one, to the best of our knowledge, was the $\Gamma$-genus coming from mirror symmetry \cite{Libgober99,HosonoKTY95a}. Furthermore, the $\hat{\Gamma}$-function turns out to play an important role in the study of multiple zeta values \cite{Cartier02,IharaKanekoZagier06}, so the results here may be of independent interest.

The plan of this note is as follows. In section 2, we recall Hirzebruch's theory of multiplicative genera and the theory of equivariant de Rham cohomology. We develop our proposal for extending the Atiyah--Witten regularization in section 3, and re-interpret this algebraically in section 4, using a formalism of Hoffman \cite{Hoffman97,Hoffman01}. In section 5, we compute the $\hat{\Gamma}$-genus for some manifolds and describe some of its properties. We conclude with a discussion in section 6.

Throughout this note, we shall assume that $M$ is a compact, connected, simply connected, oriented, smooth and finite-dimensional manifold, unless otherwise stated. In particular, this means $LM$ is assumed to be connected. We also assume that $LM$ is endowed with a topology that makes it an infinite-dimensional smooth Fr\'echet manifold.

\section{Preliminaries}

We devote this section to a rapid review of Hirzebruch's theory of multiplicative genera \cite{Hirzebruch95} and equivariant de Rham cohomology. The reader who is interested in more details about the latter theory may refer to the work of Guillemin and Sternberg \cite{GuilleminSternberg99}.

\subsection{Multiplicative Genera}

In Hirzebruch's theory, a multiplicative genus is generated by a \emph{multiplicative sequence} of polynomials $\{K_n(c_1,\ldots,c_n)\}$, where the $c_i$'s are the Chern classes of an almost complex manifold $M^{2n}$. This sequence is given by the polynomial coefficients in the series
\[K\left(\sum_{n=0}^{\infty} c_n t^n\right) = \sum_{n=0}^{\infty} K_n(c_1,\ldots,c_n) t^n,\] where $c_0 = K_0 = 1$ by convention. The multiplicative operator $K$ on the left-hand side is defined by a \emph{generating function}
 \[\phi(t) = 1 + \sum_{n=1}^{\infty} a_n t^n.\] In particular, we set $K_n(1,0,\ldots,0) = a_n$, so that $K(1+t) = \phi(t)$.

\begin{rem}
We shall preserve the traditional abuse of notation by writing $\{\phi_n(c_1,\ldots,c_n)\}$ for the polynomial $\{K_n(c_1,\ldots,c_n)\}$. This is justified by the observation that there is a one-to-one correspondence between formal power series having constant term 1 and multiplicative sequences \cite[Lemmata 1.1 and 1.2]{Hirzebruch95}. 
\end{rem}

We can then associate to a multiplicative sequence $\set{\phi_n}$ a \emph{(multiplicative) genus}, which we shall call the \emph{$\phi$-genus}. This is defined for an almost complex $2n$-manifold $M^{2n}$ by
\[\phi(M^{2n}):= \langle \phi_n(c_1,\ldots,c_n),[M^{2n}]\rangle,\] where the $c_i$'s are the Chern classes of $M^{2n}$ and $[M^{2n}]$ is the fundamental class of $M^{2n}$. The $\phi$-genus is multiplicative in the following sense (cf. \cite[Lemma 10.2.1]{Hirzebruch95}):

\begin{lem}\label{lem:genusmult}
Let $M$ and $N$ be two almost complex manifolds, and $M\times N$ be the product manifold endowed with the product almost complex structure. Then every multiplicative sequence $\set{\phi_n}$ defines a multiplicative $\phi$-genus, in the sense that \[\phi(M\times N) = \phi(M)\phi(N).\]
\end{lem} 

\begin{egs}
The \emph{$\hat{A}$-genus} of a $4n$-manifold $M$, which gives the index of the Dirac operator defined on $M$ if $M$ is spin, is generated by the $\hat{A}$-function 
\[\hat{A}(z) = \dfrac{z/2}{\sinh (z/2)}.\]
\end{egs}

\begin{egs}\label{egs:gammagen}
Let $\hat{\Gamma}(z):= e^{\gamma z}\Gamma(1+z)$ be the \emph{$\hat{\Gamma}$-function}. Recall that $1/\Gamma(1+z)$ and $1/\hat{\Gamma}(z)$ are both entire functions having power series representations with 1 as the constant term, so by Hirzebruch's theory, they also generate multiplicative genera. We shall refer to the multiplicative genus generated by $1/\hat{\Gamma}(z)$ (resp. $1/\Gamma(1+z)$) as the \emph{$\hat{\Gamma}$-genus} (respectively, the \emph{$\Gamma$-genus}, following Libgober \cite{Libgober99}).
\end{egs}

\begin{rem}
Under our standing assumptions for manifolds, we observe, following Hirzebruch \cite[p. 76]{Hirzebruch95}, that $\phi(M^{2n})$ is determined by $\phi_n(c_1,\ldots,c_n)$. Thus, in section \ref{sec:Gamma}, we shall give the value of $\hat{\Gamma}_n$, which should now be regarded as a polynomial in Chern numbers, where the $\hat{\Gamma}$-genus of $M^{2n}$ is intended.
\end{rem}

\subsection{Equivariant Cohomology and Localization}

Equivariant cohomology is usually defined with respect to the action of a group (the circle $S^1$ throughout this note). However, as we are working with differential forms, we shall use an infinitesimal model that Atiyah and Bott developed in \cite{AtiyahBott84} (though it was already implicit in the work of H. Cartan \cite{Cartan51b,Cartan51a}). Following Atiyah and Bott, we shall work with the complex numbers as our base field. 

Let $M$ be a manifold with an $S^1$-action generated by the fundamental vector field $X$. The \emph{ordinary $S^1$-equivariant (de Rham) cohomology} of $V$ is then defined to be
\[\Hy_{S^1}(M) := \Hy(\Omega_{S^1}(M)[u], d_{S^1}).\] The complex is the graded ring of polynomials in an indeterminate $u$ of degree 2 with coefficients in the $S^1$-invariant differential forms of $M$, while $d_{S^1}:= d + u\iota_X$ is the equivariant differential and $\iota_X$ is contraction with $X$.

One feature of equivariant cohomology is that it satisfies a localization theorem. This is stated in terms of the \emph{localized $S^1$-equivariant cohomology} of a manifold $M$, which is given by
\[u^{-1}\Hy_{S^1}(M):= \Hy(\Omega_{S^1}(M)[u,u^{-1}], d_{S^1}).\] 

\begin{thm}\label{thm:localize}\emph{(cf. \cite{AtiyahBott84})}
The inclusion $i\colon F\hookrightarrow M$ in $M$ of the fixed point set $F$ of the $S^1$-action induces an isomorphism $i^*\colon u^{-1}\Hy_{S^1}(V) \to u^{-1}\Hy_{S^1}(M)$ on localized $S^1$-equivariant cohomology. Since the $S^1$-action on $F$ is trivial
\[u^{-1}\Hy_{S^1}(F)\cong \Hy(F) \otimes \CC[u,u^{-1}],\] where $\Hy(F)$ is the ordinary cohomology of $F$.
\end{thm}

A consequence of this is the integration formula of Duistermaat and Heckman \cite{DuistermaatHeckman82,DuistermaatHeckman83}. This was independently derived by Berline and Vergne \cite{BerlineVergne83}, who also realized that the equivariant Euler class appears in the formula. The following theorem summarizes the results about the integration formula that we shall need in this note (cf. also \cite{BerlineGetzlerVergne92}).

\begin{thm}\label{thm:locformula}
Let $M$ be a manifold with an action of the circle $S^1$. Let $X$ be the fundamental vector field generating the $S^1$-action on $M$, $F$ be the fixed point set of the $S^1$-action with inclusion $i\colon F\hookrightarrow M$, and $\nu_F$ be the normal bundle of $F$ in $M$ such that $\nu_F$ and $F$ have compatible orientations. Let $L_{\nu_F}$ be the skew-adjoint endomorphism on $\nu_F$ induced by the $S^1$-action generated by $X$ and $R_{\nu_F}$ be the curvature of the $S^1$-invariant metric connection on $\nu_F$ induced from the Riemannian connection on $M$. Then, for a form $\alpha \in \Omega_{S^1}(M)$ that is closed under $d_{S^1}$,
\begin{equation}\label{eqn:locformula}
\int_M \alpha = \int_F  i^*(\alpha)\left[\det\left(\dfrac{L_{\nu_F}+R_{\nu_F}}{2\pi i}\right)\right]^{-1},
\end{equation} where $L_{\nu_F}$ and $R_{\nu_F}$ are considered to be complex endomorphisms when taking determinants. Furthermore, the denominator is the equivariant Euler class $e(\nu_F)$ of the normal bundle $\nu_F$.
\end{thm}

Finally, we give a construction of an equivariant differential form that represents the class $e(\nu_F)$. This is due to Jones and Petrack \cite{JonesPetrack90}.

\begin{pro}\label{lem:tauform}
With the same hypotheses as in Theorem \ref{thm:locformula}, let $\alpha$ be the differential form dual to $X$. Let $\tau\in\Omega_{S^1}(M)[u,u^{-1}]$ be the $S^1$-equivariant form given by 
\begin{equation}\label{eqn:JPtau}
\tau:= e^{-\dt \alpha},
\end{equation} $\pi\colon M\to F$ be the projection from $M$ to its fixed point set $F$, and $\pi_*\colon\Omega_{S^1}(M)[u,u^{-1}]\to \Omega(F)[u,u^{-1}]$ be integration along the fibers of $\pi$. Then,
\begin{equation}\label{eqn:inverseEuler}
\pi_*(\tau)=\left[\det\left(\dfrac{uL_{\nu_F}+R_{\nu_F}}{2\pi i}\right)\right]^{-1}.
\end{equation}
\end{pro}

\begin{rem}
It is interesting to observe that $\tau$ is a factor in the Mathai--Quillen universal Thom form. The reader is invited to compare \eref{eqn:JPtau} with formula (6.9) of \cite{MathaiQuillen86}.
\end{rem}

\begin{proof}
By construction, $\tau$ is a form closed under $d_{S^1}$. We note that, since $\alpha$ vanishes on $F$, $\tau$ satisfies the identity $i^*(\tau)=1$, where $i^*(\tau)$ is the pullback of $\tau$ by the inclusion of the fixed point set $F$ in $M$. To see that \eref{eqn:inverseEuler} holds, recall that the equivariant Thom isomorphism states that, for an equivariant form $\beta\in\Omega_{S^1}(M)[u,u^{-1}]$,
\[e(\nu_F)\pi_*(\beta) = i^*(\beta),\] where $e(\nu_F)$ is the equivariant Euler class of the normal bundle $\nu_F$ of $F$ in $M$. Since $i^*(\tau)=1$, it follows that 
\[\pi_*(\tau) = \frac{1}{e(\nu_F)}.\]
Formula \eref{eqn:inverseEuler} is then an immediate consequence of Theorem \ref{thm:locformula}.
\end{proof}

\section{Derivation of the $\hat{\Gamma}$-genus}\label{sec:derivation}

In this section, we derive the $\hat{\Gamma}$-genus. This results from an application of our proposed regularization procedure to a complex vector bundle $\pi\,\colon E\to M$ with spin structure. We also show that our regularization procedure reduces to the Atiyah--Witten regularization when $E=TM\otimes \CC$ (see Proposition \ref{pro:realcase}). Our point of departure is the paper by Jones and Petrack \cite{JonesPetrack90}, but we take a slightly different approach and work in a broader setting.

We start with a rank $m$ complex vector bundle $\pi\,\colon E \to M$, for $m\geq 2$. Note that we do not require $M$ to be a spin manifold. We endow the vector bundle with a spin structure, a smooth $S^1$-action and a $S^1$-invariant metric. Taking loops gives a \emph{rank $m$ loop bundle} (in the sense of Cohen--Stacey \cite{CohenStacey04}) $\pi_{\ell}\,\colon LE \to LM$ over $LM$, with inclusions $j\colon E\hookrightarrow LE$ and $i\colon M\hookrightarrow LM$, and $LU(m)$ as its structural group. We can then construct, by analogy with the normal bundle construction for $TM$, the \emph{$E$-normal bundle} $\nu(E)\to M$. This is defined to be $\nu(E):= i^*(LE)/E$. 

We now analyze the structure of $\nu(E)$ in more detail. Note that $\nu(E)$ inherits a complex vector bundle structure and has a Fourier decomposition
\[\nu(E)  = \bigoplus_{n=1}^{\infty} E_n,\] 
where each of the $E_n$ is a copy of $E$ with an $S^1$-action of weight $n$. There are finite-dimensional subbundles \[\nu_k(E)  = \bigoplus_{n=1}^{k} E_n\] 
with inclusions $j_k\colon \nu_k(E) \hookrightarrow \nu(E)$ into $\nu(E)$ and projections $\pi_k\colon \nu_k(E)\to M$ onto $M$. 

Within this setup, let $\tau_k$ denote the $S^1$-equivariant form on $\nu_k(E)$ as constructed in Proposition \ref{lem:tauform}. The base manifold $M$ is now the fixed point set of the $S^1$-action on $\nu_k(E)$, so we can apply Proposition \ref{lem:tauform} to see that

\begin{lem}
The equivariant cohomology class
\begin{equation}\label{eqn:tauform}
(\pi_k)_{_*}(\tau_k) = \left[\prod\limits_{n=1}^k \det\left(\dfrac{nuL_E+R_E}{2\pi i}\right)\right]^{-1}
\end{equation} is the inverse of the $S^1$-equivariant Euler class of the bundle $\nu_k(E)$. 
\end{lem}

This points to a possible definition of an $S^1$-equivariant Euler class for $\nu(E)$, but first we have to consider orientability for $\nu(E)$. From physical grounds, Witten \cite{Witten85} has argued that $LM$ is orientable if and only if $M$ is spin. Atiyah \cite{Atiyah85} and Segal \cite{Segal88} have shown that, provided $\pi_1(M)=0$, Witten's statement is true, since the obstruction to $M$ being spin transgresses to the obstruction to $LM$ being orientable. For real vector bundles, McLaughlin \cite{McLaughlin92} has proved the following:
\begin{thm}
Let $\pi_1(M)=0$ and $E\to M$ be a real vector bundle with structural group $SO(n)$, where $n\geq 4$. Then the following conditions are equivalent:
\begin{enumerate}
	\item $E\to M$ is a vector bundle with a spin structure.
	\item The structural group of the real loop bundle $LE\to LM$ can be reduced to $L_0SO(n)$, the connected component of the identity of $LSO(n)$. (This is the condition for the orientability of a loop space \cite{Segal88}.)
	\item The structural group of $LE\to LM$ has a lifting to $L\Spin(n)$.
\end{enumerate}
\end{thm}

By considering the underlying real bundle of $\pi\,\colon E\to M$, we note that the structural group condition is equivalent to requiring that $E\to M$ has structural group $U(m)\subset SO(2m)$ for $m\geq 2$, i.e. $E\to M$ has to be of rank $m\geq 2$. It then follows that  $i^*(LE)\to M$, and therefore $\nu(E)\to M$, is orientable if and only if $E\to M$ is spin.

It makes sense now to consider the \emph{$S^1$-equivariant Euler class} of the $E$-normal bundle $\nu(E)$, which we define as
\begin{equation}\label{eqn:invEuler}
e(\nu(E)):= \lim_{k\to\infty} \dfrac{1}{(\pi_k)_{_*}(\tau_k)} = \lim_{k\to\infty} \prod\limits_{n=1}^k \det\left(\dfrac{nuL_E+R_E}{2\pi i}\right).
\end{equation}
We show how this class can be written in terms of characteristic classes.

\begin{lem}\label{lem:splitEuler}
The $S^1$-equivariant Euler class of $\nu(E)$ can be re-written as 
\begin{equation}\label{eqn:invEulerprod}
e(\nu(E)) = \lim_{k\to\infty}\prod\limits_{n=1}^k \left(\dfrac{nu}{2\pi}\right)^m\cdot \lim_{k\to\infty} \left[\prod\limits_{n=1}^k \prod\limits_{j=1}^m \left(1+\dfrac{2\pi x_j}{nu}\right)\right].
\end{equation}
\end{lem}

\begin{proof}
We begin with the observation that the endomorphism $L_E$ is just $i$ times the identity. Thus, we find that we can simplify as follows:
\begin{eqnarray*}
\fl \lim_{k\to\infty} \prod\limits_{n=1}^k \det\left(\dfrac{nuL_E+R_E}{2\pi i}\right) &= \lim_{k\to\infty}\prod\limits_{n=1}^k\det\left(\dfrac{nuL_E}{2\pi i}\right)\det\left(I+\dfrac{L_E^{-1}R_E}{nu}\right)\\
&= \lim_{k\to\infty} \prod\limits_{n=1}^k\left(\dfrac{nu}{2\pi}\right)^m \cdot \lim_{k\to\infty}\prod\limits_{n=1}^k\det\left(I+\dfrac{R_E}{inu}\right).
\end{eqnarray*} 

Our next step is an observation, made by Duistermaat and Heckman \cite{DuistermaatHeckman83}, that the determinant in the second product can be expressed in terms of characteristic classes. Recall that the total Chern class of a complex vector bundle $E$ may be written as 
\[c(E) = \det\left(I+\frac{R_E}{2\pi i}\right) = 1 + c_1(E) +\cdots + c_n(E).\]  By the splitting principle, this determinant can be formally factorized into the product
\[\det\left(I+\frac{R_E}{2\pi i}\right) = \prod_{j=1}^m \left(1+x_j\right),\] where the $x_j$'s are the so-called \emph{Chern roots}, i.e. the first Chern classes of the respective formal line bundles $L_j$, where we regard $E\cong \oplus_{j=1}^m L_j$ formally as a direct sum of line bundles. Applying this factorization then yields equation \eref{eqn:invEulerprod} and completes the proof of the lemma.
\end{proof}

Note that both infinite products in formula \eref{eqn:invEulerprod} are divergent. We now propose a regularization procedure for $e(\nu(E))$. The first infinite product in \eref{eqn:invEulerprod} is handled using zeta function regularization (cf. \cite{RaySinger71}), which we recall using the approach in \cite{QuineHS93,Voros87}. 

Let $\{\mu_n\}$ be a sequence of increasing nonzero numbers and 
$Z_{\mu}(s) = \sum_{n=1}^{\infty} \mu_n^{-s}$ be its \emph{associated zeta function}. The sequence is said to be \emph{zeta-regularizable} if $Z_{\mu}(s)$ has a meromorphic continuation to a half plane containing the origin, and this meromorphic continuation is analytic at the origin and has only simple poles. If $\{\mu_n\}$ is a zeta-regularizable sequence, its \emph{zeta regularized product} is defined to be 
\[\sideset{}{_z}\prod_{n=1}^{\infty} \mu_n := \exp(-Z_{\mu}'(0)).\]
It follows from the definition that if $c$ is any nonzero number, then (cf. \cite[(1)]{QuineHS93})
\[\sideset{}{_z}\prod_{n=1}^{\infty} c\mu_n = c^{Z_{\mu}(0)}\sideset{}{_z}\prod_{n=1}^{\infty} \mu_n,\]
while if $\{\mu_n\}$ is the union of two sequences $\{\mu_{1,n}\}$ and $\{\mu_{2,n}\}$ that may be reordered arbitrarily, then (cf. \cite[(2)]{QuineHS93})
\[\sideset{}{_z}\prod_{n=1}^{\infty} \mu_n = \sideset{}{_z}\prod_{n=1}^{\infty} \mu_{1,n} \sideset{}{_z}\prod_{n=1}^{\infty} \mu_{2,n}.\]
It follows that the zeta-regularization of the first infinite product in \eref{eqn:invEulerprod} is given by
\begin{equation}\label{eqn:reg1st}
\sideset{}{_z}\prod_{n=1}^{\infty} \left(\frac{nu}{2\pi}\right)^m = \left[\left(\dfrac{u}{2\pi}\right)^{\zeta(0)}\sideset{}{_z}\prod_{n=1}^{\infty} n\right]^m = \left(\dfrac{2\pi}{\sqrt{u}}\right)^m,
\end{equation}
where the associated zeta function is the Riemann zeta function, with $\zeta(0)=-\frac{1}{2}$ and $\zeta'(0) = -\log \sqrt{2\pi}$.

Next, we implement the regularization of the second infinite product in \eref{eqn:invEulerprod} using the following regularization map $\psi_{\reg}$. Define $\psi_{\reg}\colon \Hy(M)[u,u^{-1}] \to \Hy(M)[u,u^{-1}]$ to be the operator given by extending the map
\[(1+A) \mapsto (1+A)e^{-A}\] multiplicatively to a finite product of factors of this form. Here, $A$ is a linear rational expression in terms of the Chern roots and the indeterminate $u$. 

\begin{rem}\label{rem:uniform}
We remark that in the language of Weierstrass's theory of infinite products, what $\psi_{\reg}$ achieves is to append a convergence factor to each factor of the form $(1+A)$. In particular, when the number of factors tends to infinity, the resultant infinite product is well-known to be uniformly convergent in every bounded set.
\end{rem}

We then define the \emph{regularized $S^1$-equivariant Euler class} of $\nu(E)$ to be
\begin{equation}\label{eqn:reginvEuler}
e_{\reg}(\nu(E)) := \sideset{}{_z}\prod\limits_{n=1}^{\infty} \left(\dfrac{nu}{2\pi}\right)^m \cdot \lim_{k\to\infty} \psi_{\reg} \left[\prod\limits_{n=1}^k \prod\limits_{j=1}^m \left(1+\dfrac{2\pi x_j}{nu}\right)\right].
\end{equation}

\begin{pro}\label{pro:invEuler}
The regularized equivariant Euler class of $\nu(E)$ evaluates to
\begin{equation}\label{eqn:evalreginvEuler}
e_{\reg}(\nu(E)) = \left(\dfrac{2\pi}{\sqrt{u}}\right)^m \prod_{j=1}^m \left[\hat{\Gamma}\left(\dfrac{2\pi x_j}{u}\right)\right]^{-1}.
\end{equation}
\end{pro}

\begin{proof}
Observe that $\psi_{\reg}$ acts on the second product to give
\begin{equation}\label{eqn:regaction}
\psi_{\reg} \left[\prod\limits_{n=1}^k \prod\limits_{j=1}^m \left(1+\dfrac{2\pi x_j}{nu}\right)\right] = \prod\limits_{n=1}^k \prod\limits_{j=1}^m \left[\left(1+\dfrac{2\pi x_j}{nu}\right)e^{-2\pi x_j/nu}\right].
\end{equation}
It follows from Remark \ref{rem:uniform}, together with \eref{eqn:reg1st} and \eref{eqn:regaction}, that
\begin{eqnarray*}
e_{\reg}(\nu(E)) &= \left(\dfrac{2\pi}{\sqrt{u}}\right)^m \prod\limits_{n=1}^{\infty} \prod\limits_{j=1}^m \left(1+\dfrac{2\pi x_j}{nu}\right) e^{-2\pi x_j/nu}\\
&= \left(\dfrac{2\pi}{\sqrt{u}}\right)^m \prod\limits_{j=1}^m \left[\hat{\Gamma}\left(\dfrac{2\pi x_j}{u}\right)\right]^{-1}.
\end{eqnarray*} This completes the proof.
\end{proof}

\begin{rem}
The reader may observe that the form of the regularized product \eref{eqn:reginvEuler} is closely related to the functional determinant of Voros. In fact, it is the product of the functional determinant with a non-constant exponential factor that was calculated in \cite{QuineHS93,Voros87}.
\end{rem}

Finally, we show how our proposed regularization behaves when $E=\eta\otimes \CC$ is the complexification of a real rank $m$ vector bundle $\pi_R\colon \eta\to M$. Note that since $E$ is now the complexification of a real vector bundle, $R_E$ is skew-symmetric, so that
\[c(E) = \det\left(I+\frac{R_E}{2\pi i}\right) = \det\left(I-\frac{R_E}{2\pi i}\right).\]
In particular, since we are working over the complex numbers, the odd Chern classes vanish. Observe also that $c(E)$ can now be formally factorized into
\[c(E) = \prod_{j=1}^{\lfloor m/2 \rfloor} \left(1+x_j \right)\left(1-x_j \right),\]
where the $x_j$'s are the Chern roots coming from the formal splitting of $E$ described in Lemma \ref{lem:splitEuler}. The $S^1$-equivariant Euler class of $\nu(E)$ is then given by the formula
\begin{equation}\label{eqn:zetaregdreal}
\fl e(\nu(E)) = \lim_{k\to\infty}\prod\limits_{n=1}^k \left(\dfrac{nu}{2\pi}\right)^m \cdot \lim_{k\to\infty}\left[\prod\limits_{n=1}^k \prod\limits_{j=1}^{\lfloor m/2 \rfloor} \left(1+\dfrac{2\pi x_j}{inu}\right)\left(1- \dfrac{2\pi x_j}{inu}\right)\right].
\end{equation}
The regularization procedure in this case then defines $e_{\reg}(\nu(E))$ to be
\begin{equation}\label{eqn:zetarealregd}
\fl e_{\reg}(\nu(E)) := \sideset{}{_z}\prod\limits_{n=1}^{\infty} \left(\dfrac{nu}{2\pi}\right)^m \cdot \lim_{k\to\infty} \psi_{\reg} \left[\prod\limits_{n=1}^k \prod\limits_{j=1}^{\lfloor m/2 \rfloor} \left(1+\dfrac{2\pi x_j}{inu}\right)\left(1- \dfrac{2\pi x_j}{inu}\right)\right].
\end{equation}

\begin{pro}\label{pro:realcase}
Let $\pi\colon E\to M$ be the complexification $E=\eta\otimes \CC$ of a real rank $m$ vector bundle $\eta$ over $M$, such that $E$ has a spin structure. Then the regularized $S^1$-equivariant Euler class of $\nu(E)$ evaluates to
\begin{equation}\label{eqn:realEulerregd}
e_{\reg}(\nu(E))  = \left(\dfrac{2\pi}{\sqrt{u}}\right)^m \prod_{j=1}^{\lfloor m/2 \rfloor} \left[\hat{A}\left(\dfrac{4\pi^2 x_j}{u}\right)\right]^{-1}.
\end{equation} In particular, if $\eta = TM$ is the tangent bundle of $M$, then our regularization procedure reduces to the Atiyah--Witten regularization, up to scaling of the $\hat{A}$-genus.
\end{pro}

\begin{proof}
We consider the action of the map $\psi_{\reg}$ on the product in \eref{eqn:zetarealregd}. Observe that 
\begin{eqnarray*}
\fl \psi_{\reg}\left[\prod\limits_{n=1}^k \prod\limits_{j=1}^{\lfloor m/2 \rfloor} \left(1+\dfrac{2\pi x_j}{inu}\right)\left(1- \dfrac{2\pi x_j}{inu}\right)\right] 
&= &\prod\limits_{n=1}^k \prod\limits_{j=1}^{\lfloor m/2 \rfloor} \left[\left(1+\dfrac{2\pi x_j}{inu}\right)e^{-2\pi x_j/inu}\right.\\
&& \left.\left(1- \dfrac{2\pi x_j}{inu}\right)e^{2\pi x_j/inu}\right]\\
&= &\prod\limits_{n=1}^k \prod\limits_{j=1}^{\lfloor m/2 \rfloor} \left[1+\left(\dfrac{2\pi x_j}{nu}\right)^2 \right].
\end{eqnarray*}
Note that $\sinh (2\pi^2 x/u)/(2\pi^2 x/u) = \prod_{n=1}^{\infty} \left[1+ 4\pi^2 x^2 / (n^2u^2) \right]$. It follows that the regularized $S^1$-equivariant Euler class is given by
\[
\fl e_{\reg}(\nu(E)) = \left(\dfrac{2\pi}{\sqrt{u}}\right)^m \prod_{j=1}^{\lfloor m/2 \rfloor} \dfrac{\sinh (2\pi^2 x_j/u)}{2\pi^2 x_j/u} = \left(\dfrac{2\pi}{\sqrt{u}}\right)^m \prod_{j=1}^{\lfloor m/2 \rfloor}\left[\hat{A}\left(\dfrac{4\pi^2 x_j}{u}\right)\right]^{-1}.
\]
In particular, if $\eta = TM$ is the tangent bundle of $M$, then the evaluation of $e_{\reg}(\nu(E))$ against the fundamental class of $M$ gives the inverse of the $\hat{A}$-genus of $M$, up to normalization.
\end{proof}

\section{Multiple Zeta Values and an Algebraic Formalism}\label{sec:MZVs}

In this section, we describe an algebraic formalism, developed by Hoffman \cite{Hoffman97} in his study of multiple zeta values (MZVs), that allows us to give an alternative interpretation of the map $\psi_{\reg}$ in our proposed regularization of the inverse equivariant Euler class.

To set up Hoffman's formalism, we first recall some basic theory of symmetric functions \cite{Macdonald79}. Recall that the elementary symmetric polynomials $\{e_i\}$ are generated by the function
\[E(t) = \prod\limits_{n=1}^{\infty}(1+x_n t) = \sum_{i=0}^{\infty}e_i t^i,\] while the power sum symmetric polynomials $\{p_i\}$ are generated by the function
\[P(t) = \sum\limits_{n=1}^{\infty}\frac{d}{dt}\log (1-x_n t)^{-1} = \sum_{i=1}^{\infty}p_i t^{i-1}.\] It is clear that these two functions satisfy the following relation
\begin{equation}\label{eqn:P&E_rel}
P(t)= \frac{d}{dt} \log E(-t)^{-1}.
\end{equation} We shall also need the monomial symmetric polynomials $\{m_{\lambda}\}$, where $\lambda = (\lambda_1,\lambda_2,\ldots)$ is a partition of an integer $n>0$, i.e. a sequence of numbers $\lambda_1>\lambda_2>\cdots$ with finitely many nonzero entries such that $\sum_{i=1}^{\infty} \lambda_i = n$. Note that each of these collections of symmetric polynomials form a basis for $Sym$, the algebra of symmetric functions in infinitely many variables.

In his study of MZVs, Hoffman \cite{Hoffman97} has defined a homomorphism $Z\colon Sym \to \RR$, such that on the power sum symmetric polynomials $p_i$,
\[Z(p_1)=\gamma, \quad Z(p_i) = \zeta(i) \textrm{ for } i\geq 2.\]
In particular, $Z$ acts on the generating function $P(t)$ to give
\[Z(P(t)) = \gamma + \sum_{i=2}^{\infty} \zeta(i)t^{i-1} = -\psi (1-t),\] where $\psi(z)$ is the logarithmic derivative of $\Gamma(z)$. It follows from \eref{eqn:P&E_rel} that
\[Z(E(t)) = \frac{1}{\Gamma(1+t)}\]

We now observe that a similar map $\hat{Z}\colon Sym \to \RR$ can be defined to yield the $\hat{\Gamma}$-function. Essentially, $\hat{Z}$ is a truncated version of $Z$ and acts on the power sum symmetric polynomials in the following way:
\[\hat{Z}(p_1)=0, \quad \hat{Z}(p_i) = \zeta(i) \textrm{ for } i\geq 2.\]
It follows that
\begin{equation}\label{eqn:HoffGammahat}
\hat{Z}(E(t)) = \frac{1}{\hat{\Gamma}(t)}.
\end{equation}
We use this formalism to deduce the following

\begin{pro}
Let $E$ be a complex vector bundle over $M$ and $x$ be one of its Chern roots. Let $\psi_{\reg}$ be the regularization map defined in section \ref{sec:derivation}. Then the following identity holds:
\begin{eqnarray*}
\lim_{k\to\infty} \psi_{\reg}\left(\prod\limits_{n=1}^k \left(1+\dfrac{2\pi x}{nu}\right)\right)  &= \hat{Z}\left(\lim_{k\to\infty} \prod\limits_{n=1}^k \left(1+\dfrac{2\pi x}{nu}\right) \right)\\
&= \left(\hat{\Gamma}\left(\frac{2\pi x}{u}\right)\right)^{-1}.
\end{eqnarray*}
\end{pro}

\begin{proof}
Recall that the left-hand side gives the infinite product expansion of $1/\hat{\Gamma}(\frac{2\pi x}{u})$ (cf. Proposition \ref{pro:invEuler}). It follows from \eref{eqn:HoffGammahat} that the right-hand side also yields the same expression.
\end{proof}

We now state a straightforward variation of a result of Hoffman, which gives a rather elegant description of the coefficients of the multiplicative $\hat{\Gamma}$-sequence. We omit the proof, since it is identical to the one given in \cite{Hoffman01}.

\begin{pro}
Let $\lambda=(\lambda_1,\lambda_2,\ldots)$ be a partition of $n$ and write $c_{\lambda}$ for the product $c_{\lambda_1}c_{\lambda_2}\cdots$ of Chern classes. Then $\hat{Z}(m_{\lambda})$ is the coefficient of $c_{\lambda}$ in the polynomial $\hat{\Gamma}_n(c_1,\ldots,c_n)$.
\end{pro}
	
\section{Some Properties of the $\hat{\Gamma}$-genus}\label{sec:Gamma}

In this section, we compute the $\hat{\Gamma}$-genus for some manifolds and give some properties of the $\hat{\Gamma}$-genus. We note that we only need $M$ to be an almost complex manifold for $\hat{\Gamma}(M)$ to be well-defined, so we will not require $M$ to be spin in this section.

We begin by listing, in Table \ref{tab:Gammaseq}, the first few polynomials of the multiplicative sequence $\{\hat{\Gamma}_n\}$. These are computed using the algorithm described by Libgober and Wood \cite{LibgoberWood90}, who refined it from a brief description given by Hirzebruch \cite{Hirzebruch95}.

\begin{table}
\caption{\label{tab:Gammaseq}The first few polynomials of the $\hat{\Gamma}$-sequence.}
\begin{tabular}{rp{0.85\textwidth}} 
\br
$n$ & $\hat{\Gamma}_n$\\
\mr
1 & 0\\
2 & $-\frac{1}{2}\zeta(2)(c_1^2 - 2c_2)$\\
3 & $\frac{1}{3}\zeta(3)(c_1^3 - 3c_2c_1+3c_3)$\\
4 & $\zeta(4)(c_4 - c_3c_1) + \frac{1}{2}((\zeta(2))^2- \zeta(4))c_2^2
+ (\zeta(4) - \frac{1}{2}(\zeta(2))^2)c_2c_1^2 + (\frac{1}{8}(\zeta(2))^2 - \frac{1}{4}\zeta(4))c_1^4$\\
5 & $\zeta(5)(c_5-c_4c_1)+(\zeta(2)\zeta(3)-\zeta(5))c_3c_2+(\zeta(5)-\frac{1}{2}\zeta(2)\zeta(3))c_3c_1^2+(\zeta(5)-\zeta(2)\zeta(3))c_2^2c_1+(\frac{5}{6}\zeta(2)\zeta(3)-\zeta(5))c_2c_1^3+(\frac{1}{5}\zeta(5)-\frac{1}{6}\zeta(2)\zeta(3))c_1^5$\\
\br
\end{tabular}
\end{table}

\begin{egs}
Table \ref{tab:GammaCPn} gives the values of the $\hat{\Gamma}$-genus of $\CC\textrm{P}^n$ for small values of $n$. These are computed using the formula for the total Chern class of $\CC\textrm{P}^n$ 
\[c(\CC\textrm{P}^n) = (1+h_n)^{n+1},\] 
where $h_n\in H^2(\CC\textrm{P}^n,\ZZ)$ is a generator for the second cohomology group of $\CC\textrm{P}^n$. 

\begin{table}
\caption{\label{tab:GammaCPn}Values of $\hat{\Gamma}(\CC\textrm{P}^n)$ for $n\leq 5$.}
\begin{indented}\item[]
\begin{tabular}{rp{0.05 \textwidth}p{0.3 \textwidth}} 
\br
$n$ & &$\hat{\Gamma}(\CC\textrm{P}^n)$\\
\mr
1 & &0\\
2 & &$-\frac{3}{2}\zeta(2)h_2^2$\\
3 & &$\frac{4}{3}\zeta(3)h_3^3$\\
4 & &$\frac{105}{16}\zeta(4)h_4^4$\\
5 & &$(\frac{6}{5}\zeta(5) - 6\zeta(2)\zeta(3))h_5^5$\\
\br
\end{tabular}
\end{indented}
\end{table}
\end{egs}

\begin{egs}\label{egs:LeBrun}
Consider the product of a K3 surface with the 2-sphere $M=K3\times S^2$. Using twistor theory, LeBrun \cite{LeBrun99} has shown that on $M$, there is a family of complex structures $J_n$ parametrized by an integer $n>0$. Thus, for each $n$, we have the following Chern numbers for $M$ (cf. also \cite{BarthPV84}):
\begin{equation}\label{eqn:K3Chern}
c_1^3(M,J_n) = 0, \quad c_2c_1(M,J_n)=48n, \quad c_3(M,J_n) = 48.
\end{equation} 
For $n=1$, $M$ has the product complex structure, so that since $\hat{\Gamma}(S^2)=0$, it follows from Lemma \ref{lem:genusmult} that $\hat{\Gamma}(M,J_1) = 0$. The vanishing of the $\hat{\Gamma}$-genus for $(M,J_1)$ can also be verified by comparing \eref{eqn:K3Chern} with Table \ref{tab:Gammaseq}. However, for all other values of $n$, $\hat{\Gamma}(M,J_n)$ does not vanish, so we see that the $\hat{\Gamma}$-genus of a 6-manifold depends on the choice of its complex structure. This also shows that all the hypotheses of Lemma \ref{lem:genusmult} are needed for the lemma to be true.
\end{egs}

We give some properties of the $\hat{\Gamma}$-genus for certain almost complex manifolds.

\begin{pro}\label{pro:gamma-hat}
Let $M$ be a smooth almost complex manifold. The $\hat{\Gamma}$-genus has the following properties:
\begin{enumerate}
\item The $\hat{\Gamma}$-genus vanishes for any Riemann surface $\Sigma$. Furthermore, if $M\times \Sigma$ is a product of a Riemann surface with $M$, and has the almost complex structure induced from those of $M$ and $\Sigma$, then its $\hat{\Gamma}$-genus also vanishes.

\item The $\hat{\Gamma}$-genus is a smooth invariant for $M$ if $M$ is a 4- or 8-manifold. However, it depends on the choice of a complex structure on $M$ if $M$ is a 6-manifold and, therefore, cannot be a smooth invariant of any almost complex 12-manifold that is a product of two smooth almost complex 6-manifolds.
\end{enumerate}
\end{pro}

\begin{proof}
For (1), this follows from the vanishing of $\hat{\Gamma}_1(c_1)$ and Lemma \ref{lem:genusmult}.

For (2), we first consider the case where $M$ is a 4-manifold. In this case, we note that $c_1^2-2c_2$ is just the first Pontrjagin class, so that the $\hat{\Gamma}$-genus is a multiple of the first Pontrjagin number, which is a topological invariant of a smooth 4-manifold. 

Next, if $M$ is an 8-manifold, we observe that $\hat{\Gamma}_4(c_1,c_2,c_3,c_4)$ simplifies to
\[\hat{\Gamma}_4(c_1,c_2,c_3,c_4) = (\frac{1}{8}(\zeta(2))^2 - \frac{1}{4}\zeta(4))p_1^2 + \frac{1}{2}\zeta(4)p_2.\] Thus, the $\hat{\Gamma}$-genus is again a linear combination of Pontrjagin numbers, and therefore a smooth invariant, for an 8-manifold. 

If $M$ is a smooth 6-manifold, however, Example \ref{egs:LeBrun} shows that none of the Chern numbers, except for $c_3$, can be a smooth invariant of $M$. Hence, $\hat{\Gamma}_3(c_1,c_2,c_3)$ cannot be a smooth invariant, since it is a polynomial in terms of all three Chern numbers. Since the $\hat{\Gamma}$-genus is multiplicative, it cannot therefore be a smooth invariant for a 12-manifold that is a product of two 6-manifolds.
\end{proof}

\section{Discussion}\label{sec:Discussion}

In this note, we saw that we can extend the Atiyah--Witten regularization of the $S^1$-equivariant Euler characteristic to the case of complex bundles other than $TM\otimes \CC$. As a result, the $\hat{\Gamma}$-genus, which appears to be a new multiplicative genus, was derived.

Surprisingly, the $\hat{\Gamma}$-genus and the $\hat{\Gamma}$-function turn out to have many connections to other fields of mathematics and physics. One such connection comes via a related genus, the $\Gamma$-genus of Libgober \cite{Libgober99}, whose work generalized that of Hosono et al. \cite{HosonoKTY95a} in the study of the mirror symmetry of Calabi-Yau hypersurfaces. In comparing the generating functions of these two genera under Hoffman's formalism, we have seen that the $\hat{\Gamma}$-genus is simply the $\Gamma$-genus with the Euler constant term removed. After this note was submitted, Katzarkov, Kontsevich and Pantev \cite{Katzarkov08} have introduced the $\hat{\Gamma}$-class in their new work on \textbf{nc}-Hodge theory in mirror symmetry. We note that this is different from the $\hat{\Gamma}$-genus introduced here and is, in fact, essentially Libgober's $\Gamma$-genus.

There is a further connection to the study of MZVs besides Hoffman's formalism, however. It turns out, surprisingly, that the $\hat{\Gamma}$-function appears in the context of a regularization formula that is used to recover ``missing'' relations between MZVs (see the works of Cartier \cite{Cartier02} and Ihara, Kaneko and Zagier \cite{IharaKanekoZagier06}.) This parallels its appearance here --- as a result of our proposed regularization procedure --- in the guise of the $\hat{\Gamma}$-genus. 

All of these may perhaps shed a little light on some speculative remarks of Kontsevich \cite[\S 4.6]{Kontsevich99} and Morava \cite{Morava07}. Kontsevich has argued that the functions $\hat{A}(z)$ and $\hat{\Gamma}(z)$ lie in the same orbit of the action of the Grothendieck--Teichm\"uller group on deformation quantizations, while Morava has proposed a context in algebraic topology for the appearance of the $\Gamma$-genus. The results in this note may point towards more evidence for these conjectural remarks.

\ack It is a pleasure for the author to thank his supervisors: Prof. V. Mathai, for suggesting the initial problem and drawing the author's attention to some useful references, and for his patient guidance and support; and Dr. N. Buchdahl, for encouraging a further investigation of the $\hat{\Gamma}$-genus and for his helpful advice and support.

An early version of this work, which forms part of the author's thesis, was presented at the 2007 ICE-EM Graduate School, held at the University of Queensland in July 2007, where the author benefited from a useful discussion with Prof. N. Wallach. Thanks are also due to the organizers of the 51st Annual Meeting of the AustMS, held at La Trobe University, and the Workshop on Geometry and Integrability, held at the University of Melbourne, for the opportunities to present this work.

The author is grateful to the referees for their comments and feedback, and for drawing his attention to the preprint \cite{Katzarkov08}. Finally, thanks are due to R. Green and D. M. Roberts for useful comments and conversations.

\section*{References}
\bibliographystyle{iopart}
\bibliography{workbib}

\end{document}